\documentclass{article}
\usepackage{dcolumn}
\usepackage{bm}
\usepackage{verbatim}       

\usepackage[dvips]{graphicx}
\usepackage{amsthm}
\usepackage{amssymb}
\usepackage{bm}
\usepackage{amstext}
\usepackage{amsmath}
\usepackage{amsfonts}

\newfont{\bb}{msbm10 at 12pt}

\newcommand{\dd}{{\rm d}}
\newcommand{\bd}{\begin{definition}}                
\newcommand{\ed}{\end{definition}}                  
\newcommand{\bc}{\begin{corollary}}                 
\newcommand{\ec}{\end{corollary}}                   
\newcommand{\bl}{\begin{lemma}}                     
\newcommand{\el}{\end{lemma}}                       
\newcommand{\bp}{\begin{proposition}}            
\newcommand{\ep}{\end{proposition}}                
\newcommand{\bere}{\begin{remark}}                  
\newcommand{\ere}{\end{remark}}                     

\newcommand{\bt}{\begin{theorem}}
\newcommand{\et}{\end{theorem}}

\newcommand{\be}{\begin{equation}}
\newcommand{\ee}{\end{equation}}

\newcommand{\bit}{\begin{itemize}}
\newcommand{\eit}{\end{itemize}}
\newtheorem{theorem}{Theorem}[section]
\newtheorem{corollary}[theorem]{Corollary}
\newtheorem{lemma}[theorem]{Lemma}
\newtheorem{proposition}[theorem]{Proposition}
\theoremstyle{definition}
\newtheorem{definition}[theorem]{Definition}
\theoremstyle{remark}
\newtheorem{remark}[theorem]{Remark}

\begin{document}
%

\title{Limit curve theorems in Lorentzian geometry}

\author{E. Minguzzi \footnote{Dipartimento di Matematica Applicata, Universit\`a degli Studi di Firenze,  Via
S. Marta 3,  I-50139 Firenze, Italy. E-mail:
ettore.minguzzi@unifi.it}}

\date{}
\maketitle

\begin{abstract}
The subject of limit curve theorems in Lorentzian geometry is
reviewed. A general limit curve theorem is formulated which includes
the case of converging curves with endpoints and the case in which
the limit points assigned since the beginning are one, two or at
most denumerable. Some applications are considered. It is proved
that in chronological spacetimes, strong causality is either
everywhere verified or everywhere violated on  maximizing lightlike
segments with open domain.  As a consequence, if in a chronological
spacetime two distinct lightlike lines intersect each other then
strong causality holds at their points. Finally, it is proved that
two distinct components of the chronology violating set have
disjoint closures or there is a lightlike line passing through each
point of the intersection of the corresponding boundaries.

\end{abstract}


\section{Introduction}

The limit curve theorems are surely one of the most fundamental
tools of Lorentzian geometry. Their importance  is  certainly
superior to that of analogous results in Riemannian geometry because
in Lorentzian manifolds the curves may have a causal character, and
hence it is particularly important to establish whether two points
can be connected  by a causal, a timelike or a lightlike curve.

The limit curve methods are so powerful and their range of
applicability is so wide that often the application of a limit curve
argument comes as the very first step in order to reach a desired
result. In some sense the application of a limit curve theorem may
be called a ``brute force method'', a method which sometimes can be
replaced by more elegant arguments but whose effectiveness can
hardly be denied.

The proofs of this kind of results is often lengthy, and for this
reason it is important to have them stated in a sufficiently general
and informative way. Otherwise, the risk for the researcher is that
of being forced to rebuild a slightly more general statement, all
over again, any time a modification or an improvement is needed.
Unfortunately, in my opinion, the limit curve theorems have not been
stated with sufficient generality and as a researcher I have indeed
experienced the above problem. The basic results so far available on
limit curves are scattered across different books and research
articles, with versions that rely on different conventions.
Moreover, and most importantly, the statements of those results do
not take full advantage of the powerful methods used in the proofs
so that there is in fact enough room for interesting improvements.

The aim of this work is to comment and make some order on the
results that have appeared in the literature, and to produce a
version which should be able to capture most of the information that
a limit curve theorem should give. In this way my hope is to make a
service to those researchers who use limit curve theorems in
Lorentzian geometry, and who want to rely on a general
 result with a detailed proof.

%
%

The changes experienced by the limit curve theorems in the last
decades are  worth knowing. I give a brief account which may help to
understand in which sense the version given in this work is stronger
or includes the previous formulations. I will translate the
different versions in the  notations of this work. Some
technicalities will clarify in what follows. Note that the curves
considered are always future-directed so that this adjective is
omitted throughout the work.

A first version of limit curve theorem is theorem 6.2.1 of Hawking
and Ellis \cite{hawking73}

\begin{quote}
{\em Let $\gamma_n$ be an infinite sequence of continuous causal
curves which are (past, future) inextendible. If $x$ is a limit
point of $\gamma_n$, then through $x$ there is continuous causal
curve which is (resp. past, future) inextendible and which is a
limit curve of the $\gamma_n$.}
\end{quote}

This formulation has some weak points that I am going to comment.
\begin{itemize}
\item[(i)] It uses a weak version of ``limit curve'' definition.
\item[(ii)] The convergence obtained does not allow to apply
results  on the upper semi-continuity of the length functional
unless strong causality is added.
\item[(iii)] It does not include the case of curves with both
endpoints, nor it includes the case in which the limit event ($x$
above) is not unique.

\end{itemize}

The first weak point comes from the particular definition of limit
curve used in \cite[Sect. 6.2]{hawking73}. They define $\sigma$ to
be a limit curve of $\sigma_n$ if there is a (distinguishing)
subsequence $\sigma_m$ such that for any $x\in \sigma$, every
neighborhood of $x$ intersects {\em an infinite number}  of
$\sigma_m$ ($x$ is distinguished by $\sigma_m$).

In Beem et al. \cite[Def. 3.28]{beem96} a different definition is
given where {\em an infinite number} is replaced by {\em  all but a
finite number}. There are simple examples of limit curves according
to the definition of \cite{hawking73} which are not limit curves
according to \cite{beem96}. Thus the limit curve theorem by Beem et
al. is stronger than that by Hawking and Ellis. Also, the theorem
\cite[Prop. 3.34]{beem96} on the almost equivalence between the
limit curve convergence and the $C^0$ convergence in strongly causal
spacetimes does not hold with the definition of limit curve given in
\cite{hawking73}.

Although the version given in Beem et al. \cite[Prop. 3.31]{beem96}
solved the problem (i), the formulation was pretty much similar to
that by Hawking and Ellis. In particular in applications one often
has to deal with limit curves situations in which one would like to
apply the upper semi-continuity for the length functional. It was
known, see Penrose's book \cite{penrose72}, that though the limit
curve convergence was not enough in order to guarantee the upper
semi-continuity of the length functional, at least under strong
causality  the $C^{0}$ convergence was indeed sufficient. Moreover,
it was known that in strongly causal spacetimes the $C^{0}$
convergence is actually almost equivalent to the limit curve
convergence in a sense clarified by Beem et al. in \cite[Prop.
3.34]{beem96}. In order to use the upper semi-continuity of the
length functional in a limit curve theorem application one had then
to assume the strong causality of the spacetime, pass through the
$C^{0}$ convergence of the sequence, and apply the upper
semi-continuity of length with respect to $C^0$ convergence as
proved by Penrose \cite{penrose72} (Beem et al. \cite[Remark
3.35]{beem96} refer to it and to Busemann \cite{busemann67}). It was
certainly a quite involved chain of implications, and the assumption
of strong causality was a serious drawback.

Nevertheless, the proof given by Beem et al. \cite[Prop.
3.31]{beem96} contained an important improvement. By using the
Arzela's theorem, in a way analogous to what was done by Tonelli in
the direct method of the calculus of variations \cite{buttazzo98},
they were able to show that the limiting sequence parametrized with
respect to the arc length of an auxiliary complete Riemannian
metric, converges {\em uniformly} on compact subsets to a suitable
parametrized limit curve. Galloway \cite{galloway86b} noted that the
uniform convergence on compact subsets was enough in order to
guarantee the upper semi-continuity of the length functional, at
least for curves restricted to a compact domain. This observation
was of fundamental importance because from that moment on one could
apply the limit curve theorem and the upper semi-continuity of the
length functional with no need to assume additional causality
requirements. In particular, the existence arguments for lines or
rays, being based on limit maximizing sequences, strongly benefited
from this observation.

From Beem et al. proof of the limit curve theorem, to Galloway's
observation, the technique of introducing an auxiliary complete
Riemannian metric $h$ so as to parametrize the curves with respect
to $h$-length became quite standard. The case in which the sequence
is made of curves with endpoints, converging or diverging, was
somewhat left aside and, though there were some important results in
this direction (see \cite[Theorem 8.13]{beem96}, \cite[Lemma
1]{eschenburg92}), they did not appear as a single body of research
together with the results on inextendible curves. The aim of this
work is to formulate the limit curve theorem in a way sufficiently
general to serve as a solid reference for future applications. In
particular it will include the case of converging curves with
endpoints.

We refer the reader to \cite{minguzzi06c} for most of the
conventions used in this work. In particular, we denote with $(M,g)$
a $C^{r}$ spacetime (connected, time-oriented Lorentzian manifold),
$r\in \{2, \dots, \infty\}$ of arbitrary dimension $n\geq 2$ and
signature $(-,+,\dots,+)$.  Subsequences of a sequence of curves
$\sigma_n$ are denoted with the same letter but changing the index.
Thus we can say that $\sigma_k$ is a subsequence of $\sigma_n$.

In some places in order to save space and include in one single
statement many different cases,  the generic  closed interval of the
real line is denoted $[a,b]$, where $a$ can take the value $-\infty$
and $b$ can take the value $+\infty$ (thus $[0,+\infty]$ stands for
$[0,+\infty)$). At the beginning of every lemma, theorem or
definition it is clearly pointed out if this convention applies.
Otherwise, $[a,b]$ denotes the usual compact interval, while the
letter $I$ denotes the generic closed interval of the real line.

\section{Some preliminary results}

Recall that a continuous curve $\gamma: I \to M$,  is {\em causal}
if for every convex set $U$ and $t_1,t_2 \in I$, $t_1<t_2$, with
$\gamma([t_1,t_2]) \subset U$, it is $\gamma(t_1)<_U \gamma(t_2)$
(see \cite{hawking73,minguzzi06c}). A continuous causal curve can be
shown to satisfy a local Lipschitz condition \cite[Eq.
3.14]{beem96}, and hence to be almost everywhere differentiable. The
same Lipschitz condition implies, in a suitable coordinate chart,
the boundness of velocity at those points where it is defined. Note
that the causality condition implies that there can't be $t_1,t_2
\in I$, $t_1<t_2$, such that $\gamma([t_1,t_2])=p \in M$, as it
could be for a generic continuous curve.

The Lorentzian length $l(\gamma)$ of a continuous causal curve is
defined as the greatest lower bound of the lengths of the
interpolating causal geodesics \cite{penrose72}. Because of the
almost everywhere differentiability, and of the local Lipschitz
condition, this length can be calculated with the usual integral
$l(\gamma)=\int_I \sqrt{-g(\dot{x}, \dot{x})} \, \dd t$.

Introduced on $M$ a Riemannian metric $h$, the Riemannian length
$l_0(\gamma)$ of a continuous causal curve  $\gamma: I \to M$ is
defined, as usual, as the lower upper bound of the $h$-lengths of
the interpolating $h$-geodesics. Due to the almost everywhere
differentiability, and to the local Lipschitz condition, this length
can be calculated with the usual integral $l(\gamma)=\int_I \sqrt{h
(\dot{x}, \dot{x})} \, \dd t$. Since for a continuous causal curve
there is no interval $[t_1,t_2]$ such that $\gamma([t_1,t_2])=p \in
M$,  the map $s(t)=l_0(\gamma\vert_{[t_0,t]})$, $t_0,t \in I$, is
 increasing and hence
invertible. Thus any continuous causal curve can  be reparametrized
with respect to the Riemannian length with an invertible
transformation.

The Lorentzian distance $d: M\times M \to [0,+\infty]$ is defined so
that $d(x,z)$, $x,z \in M$, is the supremum over the Lorentzian
lengths of the piecewise $C^1$ causal curves connecting $x$ to $z$
(piecewise $C^1$ can be replaced with ``continuous''). Curiously, it
is quite easy to prove that the  Lorentzian distance is lower
semi-continuous \cite[Lemma 4.4]{beem96}, while the proof of the
upper semi-continuity of the length  functional is more involved. I
give a version which is particularly suitable for our purposes. It
improves the version of  \cite[theorem 7.5]{penrose72} in that the
curves of the sequence as well as the limit curve may have or may
not have endpoints (which if present do not need to be fixed) and
strong causality is not assumed (thus embodying Galloway's
observation). The $C^{0}$ convergence is replaced with the
convergence in the uniform topology, a small price to be paid in
order to get rid of the strong causality assumption.

Recall that if $h$ is a Riemannian metric on $M$ and $d_0$ is the
associated Riemannian distance then $\gamma_n: I \to M$ converges
uniformly to $\gamma:I \to M$ if for every $\epsilon >0$ there is
$N>0$, such that for $n>N$, and for every $t \in I$,
$d_0(\gamma(t),\gamma_n(t)) <\epsilon$. For the next application
this definition is too restrictive and must be generalized to the
case in which the domains  of $\gamma_n$ differ.

\begin{definition} \label{vga}
(In this definition $a_n$, $b_n$, $a$, $b$, may take an infinite
value.) Let $h$ be a Riemannian metric on $M$ and let $d_0$ be the
associated Riemannian distance. The sequence of curves $\gamma_n:
[a_n,b_n] \to M$ converges {\em $h$-uniformly} to $\gamma:[a,b] \to
M$ if $ a_n \to a$, $ b_n \to b$, and for every $\epsilon
>0$ there is $N>0$, such that for $n>N$, and for every $t \in [a,b]\cap
[a_n,b_n]$, $d_0(\gamma(t),\gamma_n(t)) <\epsilon$.

The sequence of curves $\gamma_n: [a_n,b_n] \to M$ converges {\em
$h$-uniformly on compact subsets} to $\gamma:[a,b] \to M$ if for
every compact interval $[a',b'] \subset [a,b]$,  there is a choice
of sequences $a'_n ,b'_n \in [a_n,b_n]$, $a'_n <b'_n$, such that $
a'_n \to a'$, $ b'_n \to b'$, and for any such choice
$\gamma_n\vert_{[a'_n,b'_n]}$ converges $h$-uniformly to
$\gamma\vert_{[a',b']}$.

\end{definition}

\begin{remark} \label{rem}
Clearly, if  $\gamma_n: [a_n,b_n] \to M$ converges $h$-uniformly to
$\gamma:[a,b] \to M$  then $\gamma_n$ converges $h$-uniformly on
compact subsets to $\gamma$. Conversely, if $\gamma_n: [a_n,b_n] \to
M$ converges $h$-uniformly on compact subsets to $\gamma:[a,b] \to
M$, $[a,b]$ is compact and $ a_n \to a$, $ b_n \to b$, then
$\gamma_n$ converges $h$-uniformly to $\gamma$.
\end{remark}

\begin{remark}
Actually, the $h$-uniform convergence on compact subsets is
independent of the Riemannian metric $h$ chosen. The reason is that
if the domain of $\gamma:[a,b] \to M$ is compact then  the same is
true for its image and it is possible to find a open set of compact
closure $A$, containing $\gamma([a,b])$. Then on $\bar{A}$, given a
different Riemannian metric $h'$, there are constants $m$ and $M$
such that $m h' <h <M h'$.
\end{remark}

\begin{theorem} \label{ups}
Let $\gamma:[a,b] \to M$,  be a continuous causal curve in the
spacetime $(M,g)$ and let $h$ be a Riemannian metric on $M$.

\begin{itemize}
\item[(a)] If the sequence of continuous causal
curves $\gamma_n: [a,b] \to M$  converges $h$-uniformly to $\gamma$,
then $\limsup l(\gamma_n) \le l(\gamma) $.
\item[(b)] If the sequence of continuous causal
curves $\gamma_n: [a_n,b_n] \to M$  converges $h$-uniformly to
$\gamma$ and the curves $\gamma_n$ are parametrized with respect to
$h$-length, then $\limsup l(\gamma_n) \le l(\gamma) $. Moreover,
$\gamma_n$ converges to $\gamma$ in the $C^{0}$ topology and for
every sequence $t_n \in [a_n,b_n]$, $t_n \to t \in [a,b]$, it is $
\gamma_{n}(t_n) \to \gamma(t)$.
\end{itemize}
\end{theorem}

\begin{proof}

%

Proof of (a). Given $\varepsilon>0$ a partition of $[a,b]$ can be
found into intervals $[t_1,t_{i+1}]$, $1\le i \le m-1$, $t_i\in
[a,b]$, $t_1=a$, $t_m=b$, $t_i<t_{i+1}$, such that  the
interpolating geodesic $\eta$ passing through the events
$x_i=\gamma(t_i)$ has a length $l(\eta)\le l(\gamma)+\varepsilon/2$,
and there are convex sets $U_i$, $1\le i \le m-1$, covering
$\gamma$ such that $\gamma\vert_{[t_i,t_{i+1}]} \subset U_i$ (recall
the length definition). In particular $x_i \in U_{i-1}\cap U_i$.
%
%
%

For every $i$ let events $y_i,z_i \in U_{i-1}\cap U_i$ be chosen
such that $y_i\ll_{U_i} x_i \ll_{U_{i-1}} z_i$. Thanks to the
smoothness of the exponential map \cite[Lemma 5.9]{oneill83} the
Lorentzian distance $d_i: U_i\times U_i \to [0,+\infty]$ is finite
and continuous for each $i$. Thus the events $y_i,z_{i+1}\in U_i$
can be chosen close enough to $x_i$ and $x_{i+1}$ so that
$d_i(y_i,z_{i+1})<d_i(x_i,x_{i+1})+\varepsilon /(2m)$. Since the
image of $\gamma$ is compact and the convergence is uniform, it is
possible to find $N>0$, such that for $n>N$, $\gamma_{n} \subset A=
\bigcup^{m-1}_{i=1} U_i$ and $\gamma_n(t_i) \in I_{U_i}^{+}(y_i)\cap
I_{U_{i-1}}^{-}(z_i)$ for $2\le i\le m-1$, $\gamma_n(t_1) \in
I_{U_1}^{+}(y_1)$, $\gamma_n(t_m) \in I_{U_{i-1}}^{-}(z_m)$. The
curves $\gamma_n$ split into curves
$\gamma_n^{i}=\gamma_n\vert_{[t_i,t_{i+1}]}$ contained in $U_i$.
Now, note that the curve $\gamma^{i}_n$ can be considered as the
segment of a longer causal curve that connects $y_i$ to $z_{i+1}$
entirely contained in $U_i$, thus $l(\gamma^i_n)\le
d_i(y_i,z_{i+1})$. Finally,
\[
l(\gamma_n)\le \sum_{i=1}^{m-1} d(y_i,z_{i+1})\le \sum_{i=1}^{m-1}
d_i(x_i,x_{i+1})+ \frac{\varepsilon}{2}
=l(\eta)+\frac{\varepsilon}{2}  \le l(\gamma)+\varepsilon.
\]

Proof of (b).  Given a compact $C$, there is a constant $M>0$ such
that $-g<M^2 h$. Indeed, since $(M,g)$, by definition of spacetime,
is time orientable there is a global normalized timelike vector
field $v$, $g(v,v)=-1$. Let $\eta=g(\cdot, v)$ be the associated
1-form and define $g^{\perp}$ of signature $(0,+,\ldots,+)$ so that,
$g=-\eta \otimes \eta+g^{\perp}$. The metric $h'=\eta \otimes
\eta+g^{\perp}$ is Riemannian and given a vector $w$,
$-g(w,w)=\eta(w)^{2}-g^{\perp}(w,w)\le
\eta(w)^{2}+g^{\perp}(w,w)=h'(w,w)$. Since $C$ is compact there is
$M>0$ such that $h'< M^2 h$.

Let $A$ be an open set of compact closure, $\bar{A}=C$, containing
$\gamma$, and let $\Delta>0$ the Riemannian distance between
$\gamma$ and $\dot{A}$. For every $a',b'\in (a,b)$, $a'<b'$, $\vert
b-b'\vert<\Delta/2$, $\vert a-a' \vert <\Delta/2$,  it is, for
sufficiently large $n$, $[a',b']\subset [a_n,b_n]$, thus if the
curves are restricted to the interval $[a',b']$ (a) holds, $\limsup
l(\gamma_n\vert_{[a',b']})\le l(\gamma\vert_{[a',b']})$. Because of
uniform convergence for sufficiently large $n$, it holds
$d_0(\gamma,\gamma_n([a',b']))<\Delta/2$, and since $\gamma_n$ are
parametrized with respect to $h$-length $d_0(\gamma_n(a'),
\gamma_n(a)) \le \vert a-a'\vert$ and analogously for the future
endpoint. Using the triangle inequality it follows that $\gamma_n$
is entirely contained in $A$ which proves the $C^{0}$ convergence.
Also $l(\gamma_n\vert_{[a,a']}) \le {M}\,
l_0(\gamma_n\vert_{[a,a']})={M}\,  \vert a'-a\vert $, and
$l(\gamma_n\vert_{[b',b]}) \le {M}\,  \vert b-b'\vert $, so that
$l(\gamma_n)\le l(\gamma_n\vert_{[a',b']})+{M}\, \vert a'-a\vert +
{M}\, \vert b-b'\vert$, and finally $\limsup l(\gamma_n)\le \limsup
l(\gamma_n\vert_{[a',b']})+ {M}\,\vert a'-a\vert + {M}\, \vert
b-b'\vert \le l(\gamma\vert_{[a',b']})+ {M}\,\vert a'-a\vert + {M}\,
\vert b-b'\vert$. Using the arbitrariness of $a'$ and $b'$, $\limsup
l(\gamma_n) \le l(\gamma) $.

Finally, consider the sequence $t_n \in [a_n,b_n]$, $t_n \to t \in
[a,b]$. Let $\epsilon >0$, recall that $\gamma$ is continuous, and
take $t'\in (a,b)$ sufficiently close to $t$  that $d_0(\gamma(t),
\gamma(t'))<\epsilon/4$ and $\vert t'-t\vert \le \epsilon/4$. For
sufficiently large $n$, $t' \in [a_n,b_n]$, and because of
convergence we can also assume
$d_0(\gamma(t'),\gamma_n(t'))<\epsilon/4$. Moreover, if $n$ is
sufficiently large $\vert t_n-t \vert <\epsilon/4$ and finally
$d_0(\gamma(t),\gamma_n(t_n)) \le d_0(\gamma(t),
\gamma(t'))+d_0(\gamma(t'),
\gamma_n(t'))+d_0(\gamma_n(t'),\gamma_n(t_n))\le \frac{1}{2}
\epsilon+\vert t'-t_n\vert \le \epsilon $

\end{proof}

\begin{lemma} \label{mas}
Let $h$ be a Riemannian metric on the spacetime $(M,g)$. Every event
$x \in M$, admits a globally hyperbolic coordinate neighborhood $(V,
x^{\mu})$ ($x^0=const.$ are Cauchy hypersurfaces for $V$)  and a
constant $K>0$ such that if $\gamma: I \to M$ is a continuous causal
curve and $\gamma\vert_{[t_1,t_2]} \subset V$ then
$l_0(\gamma\vert_{[t_1,t_2]}) \le K \vert
x^{0}(\gamma(t_2))-x^{0}(\gamma(t_1))\vert$.
\end{lemma}

\begin{proof}
Every event $x\in M$ admits arbitrary small globally hyperbolic
neighborhoods \cite{minguzzi06c}, in particular inside a
neighborhood $U\ni x$, of compact closure. The neighborhood admits
coordinates $\{x^{\mu}\}$ so that $g^{+}=-(\dd x^{0})^{2}+\sum_i
(\dd x^{i})^{2}$, is such that $g^{+}>g$ (for details see
\cite[Lemma 2.13]{minguzzi06c}). Since $g^{+}>g$, $\gamma$ is causal
with respect to $g^{+}$.

On the compact $\bar{V}$ consider the Riemannian metric
$\tilde{h}=(\dd x^{0})^{2}+\sum_i (\dd x^{i})^{2}$, then there is a
constant $M>0$ such that $h < M^2 \tilde{h}$ on $\bar{V}$. Let
$t_2>t_1$ such that $\gamma([t_1,t_2]) \subset V$. The
$\tilde{h}$-length of $\gamma\vert_{[t_1,t_2]}$  is the supremum of
the lengths of the interpolating piecewise $C^1$
$\tilde{h}$-geodesics, which for sufficiently fine interpolation are
necessarily $g^+$-causal. Using the condition of $g^{+}$-causality,
calling $\sigma$ one of the interpolating geodesics of
$\gamma\vert_{[t_1,t_2]}$, it is easily seen that
$\tilde{l}_0(\sigma)\le \sqrt{2} \,\vert
x^{0}(\sigma(t_2))-x^{0}(\sigma(t_1)) \vert$, and taking the
supremum over the interpolating geodesics, since the endpoints
remain the same, $\tilde{l}_0(\gamma\vert_{[t_1,t_2]})\le \sqrt{2}\,
\vert x^{0}(\gamma(t_2))-x^{0}(\gamma(t_1)) \vert$, thus
$K=\sqrt{2}\, M$.

\end{proof}

The next lemma had been proved, in one direction, in \cite[Lemma
3.65]{beem96} and in the other direction at the end of the proof of
\cite[Prop. 3.31]{beem96}. This last step is given here a different,
shorter proof.

\begin{lemma} \label{nsf}
Let $(M,g)$ be a spacetime and let $h$ be a complete Riemannian
metric on $M$. A continuous causal curve $\gamma$ once parametrized
with respect to $h$-length has a domain unbounded from above  iff
future inextendible and unbounded from below iff past inextendible.

\end{lemma}

\begin{proof}
Let $(a,b)$ be the interior of a domain obtained by parametrizing
the curve with respect to $h$-length, with possibly $b=+\infty$ and
$a=-\infty$. Assume $\gamma$ future inextendible and let $p=
\gamma(t)$, $t \in (a,b)$, and consider the balls
$B_n(p)=\{q:d_0(p,q)\le n\}$. They are compact because of the
Hopf-Rinow theorem. If $\gamma\vert_{[t,b)}$ is not entirely
contained in $B_n(p)$ for a certain $n$, then
$b-t=l_0(\gamma\vert_{[t,b)}) \ge n$ for all $n$, thus $b=+\infty$.
Otherwise, $\gamma\vert_{[t,b)}$ is contained in a compact and there
is a sequence $t_k \in (a,b)$, $t_k \to b$, such that $\gamma(t_k)
\to q$. But since $q$ can't be a limit point there are $\bar{t}_k
\in (a,b)$, $\bar{t}_k \to b$, such that $\gamma(\bar{t}_k) \notin
B_{1/n}(q)$ for a certain $n$. For sufficiently large $k$,
$\gamma(t_k) \in B_{1/(2n)}(q)$, and hence $\gamma\vert_{[t,b)}$
enters $B_{1/(2n)}(q)$ and escapes $B_{1/n}(q)$ infinitely often,
and thus has infinite length, $b-t=l_0(\gamma\vert_{[t,b)})=+\infty$
$\Rightarrow b=+\infty$.

Assume $b=+\infty$ then if $\gamma$ has a future endpoint $x$ there
is a globally hyperbolic coordinate neighborhood $V$, as given in
lemma \ref{mas}, and a constant $t \in (a,+\infty)$ such that
$\gamma\vert_{[t,+\infty)} \subset V$. But there is also a constant
$K>0$ such that for $n>t$, $1=l_0(\gamma\vert_{[n,n+1]}) \le K \vert
x^{0}(\gamma(n+1))-x^{0}(\gamma(n))\vert$, thus it is impossible
that $x^{0}(\gamma(n)) \to x^{0}(x)$, and hence that $\gamma(n) \to
x$.
\end{proof}

%

The proof of the next lemma is in part contained in \cite[Prop.
3.31]{beem96}, however the original proof contained a  gap that is
fixed here.

\begin{lemma} \label{mae}
Let $(M,g)$ be a spacetime and let $h$ be a Riemannian metric on
$M$. If the continuous causal curves $\gamma_n: I_n \to M$
parametrized with respect to $h$-length converge $h$-uniformly on
compact subsets to $\gamma: I \to M$ then $\gamma$ is a continuous
causal curve.
\end{lemma}

\begin{proof}
Let $\gamma([t_1,t_2]) \subset U$, $[t_1,t_2] \subset I$, with $U$
($g$-)convex neighborhood. Let $\Delta>0$ be the Riemannian distance
between the compact $\gamma([t_1,t_2])$ and the closed set $U^{C}$.
By uniform convergence on compact subsets there are sequences
$t_{1n}, t_{2n} \in I_n$, such that $t_{1n} \to t_1$, $t_{2n} \to
t_2$, and $\gamma_{n} \vert_{[t_{1n}, t_{2n}]}$ converges
$h$-uniformly to $\gamma\vert_{[t_1,t_2]}$, in particular for large
$n$, $\gamma_{n} \vert_{[t_{1n}, t_{2n}]}$ has an image included in
$U$. Thus, $\gamma_{n}(t_{1n}) \le_{U}  \gamma_{n}(t_{2n})$, and
because of  theorem \ref{ups}, case (b),  $\gamma_{n}(t_{1n})\to
\gamma(t_1)$, $\gamma_{n}(t_{2n})\to \gamma(t_2)$.  Now, recall that
$J^{+}_{U}$ is closed in a convex neighborhood, and hence
$\gamma(t_{1}) \le_{U}  \gamma(t_{2})$. It remains to prove that
$\gamma(t_{1}) \ne   \gamma(t_{2})$ so that $\gamma(t_{1}) <_{U}
\gamma(t_{2})$ (the proof  of \cite[Prop. 3.31]{beem96} lacks this
part). Indeed, if $\gamma(t_{1}) = \gamma(t_{2})$ then $\gamma(t)=
\gamma(t_{1})$, for every $t \in [t_1,t_2]$, otherwise there would
be $\bar{t} \in [t_1,t_2]$, $\gamma(t_1)<_U \gamma(\bar{t})<_U
\gamma(t_1)$ which would violate the causality of $(U,g\vert_{U})$
(recall that every convex neighborhood is causal). Finally,
\begin{align*}
t_2-t_1 &=\limsup (t_{2n}-t_{1n})=\limsup
{l}_0(\gamma_n\vert_{[t_{1n},t_{2n}]}) \\
&\le K\,\limsup \vert
x^{0}(\gamma_n(t_{2n}))-x^{0}(\gamma_n(t_{1n})) \vert =K \vert
x^{0}(\gamma(t_{2}))-x^{0}(\gamma(t_{1})) \vert
\end{align*}
thus if $t_2 \ne t_1$ necessarily $\gamma(t_2) \ne \gamma(t_1)$.

\end{proof}

With slight modifications the next local result is contained in
Lemma 1 of \cite{beem76} (see also  \cite[Lemma 3.1]{aguirre89}).
$J^{+}_S$ denotes the Seifert causal relation
\cite{seifert71,minguzzi07}.

\begin{lemma} \label{kas}
Let $(M,g)$ be a spacetime, and $p \in M$. Let $g_n\ge g$ be a
sequence of metrics such that $g_{n+1}\le g_n$, and assume that the
metrics $g_n$, regarded as sections of the bundle $T^*M\otimes T^*M
\to M$, converge pointwisely to the metric $g$. There is a
$g$-convex  neighborhood $V\ni p$, contained, for all $n$ in
$g_n$-convex neighborhoods $V_n$, such that if $(x_n,z_n) \in
J^{+}_{(V,g_n)}$, and $(x_n,z_n) \to (x,z)$, then $(x,z) \in
J^{+}_{(V,g)}$. In particular $J^{+}_{S (V,g)}=J^{+}_{(V,g)}$.
\end{lemma}


\begin{corollary}
Let $(M,g)$ be a spacetime and let $h$ be a Riemannian metric on
$M$. Let $g_n\ge g$ be a sequence of metrics such that $g_{n+1}\le
g_n$, and assume that the metrics $g_n$, regarded as sections of the
bundle $T^*M\otimes T^*M \to M$, converge pointwisely to the metric
$g$. A curve which is a continuous $g_n$-causal curve for every $n$
is actually a  continuous $g$-causal curve.

In particular, let $\gamma_n: I_n \to M$ be a continuous
$g_n$-causal curve parametrized with respect to $h$-length, and
assume that the sequence $\gamma_n$ converges $h$-uniformly on
compact subsets to $\gamma: I \to M$, then $\gamma$ is a continuous
$g$-causal curve.
\end{corollary}

\begin{proof}
We have to prove that for every ($g$-)convex set $U$, and interval
$[t_1,t_2] \subset I$, such that $\gamma([t_1,t_2]) \subset U$, it
is $\gamma(t_1)<_{(U,g)}\gamma(t_2)$. To this end it is sufficient
to prove the statement with $U$ replaced with the set $V$ whose
properties are given by lemma \ref{kas}. Indeed, $\gamma([t_1,t_2])$
being compact can be covered with a finite number of such
neighborhoods contained in $U$. Thus assume that $U$ has the
properties of lemma \ref{kas}, in particular it is $g$-convex and
contained in $g_n$-convex sets $U_n$. Then
$(\gamma(t_1),\gamma(t_{2})) \in J^{+}_{(U,g_n)}$ and
$\gamma(t_1)\ne \gamma(t_2)$, because $\gamma$ is continuous
$g_n$-causal. Using the property of $U$ it follows
$(\gamma(t_1),\gamma(t_2)) \in J^{+}_{(U,g)}$, and $\gamma(t_1)\ne
\gamma(t_2)$, hence $\gamma$ is continuous $g$-causal.

For every $k>0$, the sequence $\gamma_n$ for $n>k$ converges
$h$-uniformly on compact subsets to $\gamma: I \to M$. Since all the
causal curves are $g_k$-causal, the limit curve $\gamma$ is
continuous $g_k$-causal by lemma \ref{mae}, where $k$ is arbitrary,
hence it is continuous $g$-causal by the previous observation.
\end{proof}

\begin{remark}
The previous result is particularly important in connection with
stable causality. It proves that in many cases the limit curve is
actually $g$-causal though the limiting sequence is made of
$g_n$-causal curves with $g_n\ge g$. Its main idea was successfully
applied by Beem in \cite[Theorem 2]{beem76}. It is important to keep
it in mind because, while the next limit curve theorems will be
stated using sequences of  curves which are causal with respect to
the same metric $g$, the theorems can be easily generalized to the
case contemplated by the previous lemma.
\end{remark}

\begin{definition}
A continuous causal curve $\gamma: I \to M$, is maximizing if, for
every $t_1,t_2 \subset I$, $t_1<t_2$,
$d(\gamma(t_1),\gamma(t_2))=l(\gamma\vert_{[t_1,t_2]})$.

A sequence of continuous causal curves $\gamma_n: I_n \to M$, is
limit maximizing if defined \[\epsilon_n= \sup_{t_1,t_2 \in I_n,
t_1<t_2}
[d(\gamma_n(t_1),\gamma_n(t_2))-l(\gamma_n\vert_{[t_1,t_2]}) ] \ge
0,\] it is $\lim_{n \to +\infty} \epsilon_n=0$.
\end{definition}

In particular a maximizing causal curve is a geodesic without
conjugated points, but for, possibly, the endpoints. If it is
inextendible it is called a {\em line}, if it is future inextendible
but has past endpoint it is called a future  {\em ray} (and
analogously in the past case).

The Lorentzian distance is not a conformal invariant function, as a
consequence the property of being a line or a ray for a causal curve
is not a conformal invariant property. An exception are the
lightlike lines or rays, indeed  they can be given the following
equivalent conformal invariant definition.

\begin{definition}
A {\em lightlike line} is a achronal inextendible continuous causal
curve. A {\em future lightlike  ray} is a achronal future
inextendible continuous causal curve with a past endpoint (and
analogously in the past case).
\end{definition}

The next theorem extends a result by Eschenburg and Galloway
\cite{eschenburg92} to the case of curves  without both endpoints,
their result being already an improvement with respect to
\cite[Prop. 8.2]{beem96} which used strong causality.

\begin{theorem} \label{maxs}
Let $(M,g)$ be a spacetime and let $h$ be a Riemannian metric on
$M$. If the sequence of continuous causal curves $\gamma_n: I_n \to
M$ is limit maximizing, the curves are parametrized with respect to
$h$-length and  the sequence converges $h$-uniformly on compact
subsets to the  curve $\gamma:I \to M$, then $\gamma$ is a
maximizing continuous causal curve.  Moreover, given $[a,b] \subset
I$ there are $[a_n,b_n] \subset I_n$,  such that $ a_n  \to a$, $
b_n  \to b$ and for any such choice
\begin{equation}
\lim l(\gamma_n\vert_{[a_n,b_n]})= \lim
d(\gamma_n(a_n),\gamma_n(b_n))=l(\gamma \vert_{[a,b]}) =
d(\gamma(a),\gamma(b)).
\end{equation}

\end{theorem}

\begin{proof} The curve $\gamma$ is a continuous causal curve
because of lemma \ref{mae}. Let $[a,b] \subset I$, then there are
$[a_n,b_n] \subset I_n$,  such that $ a_n  \to a$, $ b_n  \to b$,
and $\gamma_n \vert_{[a_n,b_n]}$ converges $h$-uniformly to
$\gamma_{[a,b]}$. Since $d$ satisfies the reverse triangle
inequality
$d(\gamma_n(a_n),\gamma_n(b_n))-l(\gamma_n\vert_{[a_n,b_n]})\le
\epsilon_n$, with $\epsilon_n \to 0$. Using case (b) of theorem
\ref{ups} and the lower semi-continuity of the distance
\begin{align*}
 d(\gamma(a),\gamma(b)) & \le \liminf d(\gamma_n(a_n),\gamma_n(b_n))
\le \limsup d(\gamma_n(a_n),\gamma_n(b_n)) \\
& \le \limsup l(\gamma_n\vert_{[a_n,b_n]}) \le l(\gamma
\vert_{[a,b]}) \le d(\gamma(a),\gamma(b))
\end{align*}
hence $d(\gamma(a),\gamma(b))=l(\gamma \vert_{[a,b]}) $ which
concludes the proof.
\end{proof}

\begin{remark}
 Given two converging sequences $x_n\to x$, $z_n \to z$, such that
 $x_n<z_n$ and $d(x_n,z_n)<+\infty$, it is always possible to
 construct a limit maximizing sequence of curves $\gamma_n:
 [a_n,b_n] \to M$, $x_n=\gamma(a_n)$, $z_n=\gamma(b_n)$. This
 observation is particularly useful, since the existence of a limit
 maximizing sequence is the starting point from which many results
 on the existence of causal rays or lines are obtained. Sometimes,
 the sequence of endpoints  may not satisfy $d(x_n,z_n)<+\infty$. In
 this case, provided the spacetime is strongly causal,
 it is still possible to construct a sort of limit maximizing
 sequence. The reader is referred to \cite[Chap. 8]{beem96} for details.
\end{remark}

The next lemma develops an idea  used by Eschenburg and Galloway in
\cite[Lemma 1]{eschenburg92}.

\begin{lemma} \label{vja}
Let $(M,g)$ be a spacetime and let $h$ be a complete Riemannian
metric on $M$.
\begin{itemize}
\item[(i)] If the continuous causal curve $\gamma:(a,b) \to M$ parametrized
with respect to $h$-length, is such that
$l(\gamma_{[t,b)})=+\infty$, for one (and hence every) $t \in (a,b)$
then $b=+\infty$ (and analogously in the past case).
\item[(ii)] If the
sequence of continuous causal curves $\gamma_n:(a_n,b_n) \to M$
parametrized with respect to $h$-length, is such that $l(\gamma
\vert_{[t_n,b_n)}) \to +\infty$, for a sequence $t_n \in (a_n,b_n)$,
$t_n \to t$, such that $\gamma_n(t_n) \to q \in M$, then $b_n \to
+\infty$ (and analogously in the past case).
\end{itemize}
\end{lemma}

\begin{proof}
Proof of (i). Let $p=\gamma(t)$, if $\gamma_{[t,b)}$ escapes every
ball $B_{n}(p)$ necessarily $b-t=l_0(\gamma \vert_{[t,b)})=+\infty$,
thus   $b = +\infty$. Otherwise, the image of $\gamma_{[t,b)}$ is
contained in a compact (Hopf-Rinow theorem) $C=B_{n}(p)$ for a
suitable $n$. Given the compact  $C$, there is a constant $M>0$ such
that $-g<M^2 h$ (see the proof of (b) theorem \ref{ups}), then
$b-t=l_0(\gamma \vert_{[t,b)})>M^{-1} l(\gamma
\vert_{[t,b)})=+\infty$ which implies $b = +\infty$.

Proof of (ii). Assume not then there is a subsequence $\gamma_k$
such that $b_k <N$, for a suitable $N>0$. If for every $s>0$ only a
finite number of $\gamma_k \vert_{[t_k,b_k)}$ is entirely contained
in $B_{s}(q)$ necessarily $b_k-t_k=l_0(\gamma_k \vert_{[t_k,b_k)})
\to +\infty$ which implies $b_k \to +\infty$ a contradiction. Thus
there is subsequence $\gamma_i$ of $\gamma_k$ which is entirely
contained in a compact $C=B_{s}(q)$ for a suitable $s>0$. But there
is a constant $M>0$ such that $-g<M^2 h$ on $C$, then
$b_i-t_i=l_0(\gamma \vert_{[t_i,b_i)})>M^{-1} l(\gamma_i
\vert_{[t_i,b_i)})\to +\infty$ which implies $b_i \to +\infty$,
while $b_i <N$. The overall contradiction proves that $b_n \to
+\infty$.

\end{proof}

\begin{theorem} \label{xxa}
Let $\gamma: [a,b] \to M$ be a causal curve such that
$\gamma(a)=\gamma(b)$ then either (i) $\gamma'(a) \propto
\gamma'(b)$ and $\gamma$ is obtained from the domain restriction of
a closed lightlike line $\eta$ or (ii) $\gamma$ is entirely
contained in the chronology violating set. Moreover, if $\gamma$
does not intersect the closure of the chronology violating set then
(i) holds and $\eta$ is a complete geodesic.
\end{theorem}

\begin{proof}
If $\gamma$ is not a achronal then there are points $p,q \in \gamma$
such that $p \ll q$. Take arbitrary $r \in \gamma$, since $\gamma$
is closed, $r \le p$ and $q \le r$, thus $r \le p \ll q \le r$, i.e.
$r\ll r$, that is $r$ belongs to the chronology violating set. Thus
$\gamma$ is achronal or (ii) holds. Assume (ii) does not hold. The
achronality implies that $\gamma$ is a lightlike geodesic. Note that
if $\gamma'(a) \propto \gamma'(b)$ does not hold then rounding the
corner it is possible to find $p,q \in \gamma$ such that $p \ll q$.
Thus if (ii) does not hold taking infinite rounds over $\gamma$ it
is possible to obtain an achronal inextendible causal curve i.e. a
lightlike line. The last statement follows from proposition 6.4.4 of
\cite{hawking73}.
\end{proof}

\section{The limit curve theorem}

In the following limit curve theorem the sequence $\gamma_n$ is made
of continuous causal curves parametrized with respect to $h$-length.
The statement that the curves converge $h$-uniformly on every
compact subset is a very powerful result that contains a lot of
information. It is then natural to formulate the theorem so as to
mention the role of the parametrization. Nevertheless, since the
Riemannian length functional is lower semi-continuous but non upper
semi-continuous the parametrization of the limit curve is not
necessarily the natural $h$-length parametrization.

The theorem is quite lengthy, and  at first its meaning may be
difficult to grasp. However, in applications the reader may use the
information on the converging sequence  to select a particular case
among those there considered. Case (1) and (2) can be regarded as
particular cases of case (3). However,  they are stated separately
as they have some peculiarities which are  particularly useful in
applications.

Since the theorem deals with a sequence of curves possibly with
endpoints, some additional mild requirements are required in order
to guarantee that the sequence does not shrink to a single event and
that the limit curve does indeed exist.

\begin{theorem} \label{main}
 Let $(M,g)$ be a spacetime,  and let $h$ be a complete
Riemannian metric.

\begin{itemize}
\item[(1)] [One converging sequence case] (here $a_k$, $b_k$, $a$, $b$, may take an infinite
value.) Let $y$ be an accumulation point of a sequence  of
continuous causal curves.  There is a subsequence parametrized with
respect to $h$-length, $\gamma_k: [a_k,b_k] \to M$, $0 \in
[a_k,b_k]$, such that $\gamma_k(0) \to y $ and such that the next
properties hold. There are $a\le 0$ and $b\ge 0$, such that $a_k \to
a$ and $b_k \to b$. If there is a neighborhood $U$ of $y$ such that
only a finite number of $\gamma_k$ is entirely contained in $U$
(which happens iff $b>0$ or $a<0$) then there is a continuous causal
curve $\gamma:[a,b] \to M$, $0\in [a,b]$, $y=\gamma(0)$, such that
$\gamma_k$ converges $h$-uniformly on compact subsets to $\gamma$.

In particular if $a>-\infty$, then $a_k>-\infty$ and
$x_k=\gamma_k(a_k)$ converges to $x=\gamma(a)$. Analogously, if
$b<+\infty$, then $b_k<+\infty$, and $z_k=\gamma_k(b_k)$ converges
to $z=\gamma(b)$.  If $a_k>-\infty$, $x_k=\gamma_k(a_k) \to +\infty$
or $l(\gamma_k\vert_{[a_k,0]}) \to +\infty$ then $a=-\infty$, and if
$b_k<+\infty$, $z_k=\gamma(b_k) \to +\infty$ or
$l(\gamma_k\vert_{[0,b_k]}) \to +\infty$ then $b=+\infty$.

All the mentioned parametrized curves, the sequence $\gamma_k$ and
$\gamma$, are future inextendible iff their interval of definition
is unbounded from above, and past inextendible iff their interval of
definition is unbounded from below.

Finally, if $\gamma_k$ is limit maximizing then $\gamma$ is
maximizing.
\item[(2)] [Two converging sequences case] (here $b$ may take an infinite
value.) Let $\gamma_n$ be a sequence of continuous causal curves
with past endpoint $x_n$ and future endpoint $z_n$ such that $x_n
\to x$, $z_n \to z$. There is a subsequence parametrized with
respect to $h$-length denoted $\gamma^x_k: [0,b_k] \to M$,  such
that $x_k=\gamma^x_k(0) \to x $, $z_k=\gamma^x_k(b_k) \to z $, a
analogous reparametrized sequence of continuous causal curves
$\gamma^z_k(t)=\gamma^x_k(t+b_k)$, $\gamma^z_k: [-b_k,0] \to M$,
such that $x_k=\gamma^z_k(-b_k) \to x $, $z_k=\gamma^z_k(0) \to z $,
all such  that the next properties hold. There is $b\ge 0$, such
that $b_k \to b$. If there is a neighborhood $U$ of $x$ such that
only a finite number of $\gamma_k$ is entirely contained in $U$
(which is true iff $b>0$ or if $x\ne z$ or if $x_n=z_n$) then there
is a continuous causal curve $\gamma^x:[0,b] \to M$,
$x=\gamma^x(0)$, and a continuous causal curve $\gamma^z:[-b,0] \to
M$, $z=\gamma^z(0)$, such that $\gamma^x_k$ converges $h$-uniformly
on compact subsets to $\gamma^x$, and $\gamma^z_k$ converges
$h$-uniformly on compact subsets to $\gamma^z$.

There are two cases,
\begin{itemize}
\item $0<b<+\infty$, $\gamma^z(t)=\gamma^x(t+b)$, $\gamma^x(b)=z$,
$\gamma^z(-b)=x$, so that $\gamma^x$ and $\gamma^z$ connect $x$ to
$z$, they are one the reparameterization of the other and $\limsup
l(\gamma^x_k) \le l(\gamma^x)$,
\item  $b=+\infty$, $\gamma^x$ is future inextendible, $\gamma^z$ is
past inextendible and given $t^x \in [0,+\infty)$, $t^z \in
(-\infty,0]$, there is $K(t^x,t^z)>0$ such that for $k>K$,
$\gamma_k^x(t^x) \le \gamma_k^z(t^z)$, in particular for every
choice of $\bar{x} \in \gamma^x$ and $\bar{z} \in \gamma^z$, it is
$(\bar{x},\bar{z}) \in \bar{J}^{+}$. If $\gamma^x$ is not a
lightlike ray then $\gamma^z \subset \bar{J}^{+}(x)$, and if
$\gamma^z$ is not a lightlike ray then $\gamma^x \subset
\bar{J}^{-}(z)$. If neither $\gamma^x$, nor $\gamma^z$ is a
lightlike ray then $z \in I^{+}(x)$. If ($\gamma^x$ is not a
lightlike ray or $\forall n$, $x_n=x$) and ($\gamma^z$ is not a
lightlike ray or $\forall n$, $z_n=z$) then $\limsup
l(\gamma_k^x)\le d(x,z)$. If the spacetime is non-totally
imprisoning  then the curves $\gamma_n$ are not all contained in a
compact.

\end{itemize}

If the curves $\gamma_k$ are not all contained in a compact or if
$l(\gamma_n) \to +\infty$ then $b=+\infty$.

If $x=z$ and $b=+\infty$ then $\gamma=\gamma^x \circ \gamma^z$ is an
inextendible limit (cluster) curve of $\gamma_n$, for every
$\tilde{x}, \tilde{z} \in \gamma$, it is $(\tilde{x},\tilde{z})\in
\bar{J}^{+}$ and strong causality is violated at every point of
$\gamma$. Moreover, if $\gamma$ is not a lightlike line then all but
a finite number of the curves $\gamma_k^x$ intersect the chronology
violating set.

If $x=z$ and $0<b<+\infty$ then $\gamma^x$ is a closed continuous
causal curve starting and ending at $x$. The curve $\gamma$ obtained
making infinite rounds over $\gamma^x$ is inextendible and causality
is violated at every point of $\gamma$. Moreover, either $\gamma$ is
a lightlike line or it is entirely contained in the chronology
violating set.

%

Finally whether $x=z$ or not, if $\gamma_n$ is limit maximizing then
both $\gamma^x$ and $\gamma^z$ are maximizing and if, moreover,
$b=+\infty$ then $l(\gamma^x)+l(\gamma^z) \le \limsup l(\gamma_n)$.
Thus, in this last case if $\limsup l(\gamma_n)<+\infty$ then
$\gamma^x$ and $\gamma^z$ are lightlike rays or one of them is an
incomplete timelike ray.

\item[(3)] [general case] (here $a$, $b$, $a_n$, $b_n$, $a_k$, $b_k$, $a^{(i)}$, $b^{(i)}$, may take an infinite
value.) Let $\gamma_n:[a_n,b_n] \to M$ be a sequence of continuous
causal curves parametrized with respect to $h$-length. Let
$\{y^{(i)}\}$, $\{i\}=G\subset \mathbb{N}$, be a non-empty and at
most numerable set of limit points for $\gamma_n$ (namely, every
neighborhood of $y^{(i)}$ intersect all but a finite number of the
curves $\gamma_n$). If $\limsup(b_n-a_n)>0$ then there is a
subsequence $\gamma_k$ of $\gamma_n$ such that $\lim (b_k-a_k)$
exists, there are sequences $t^{(i)}_k \in [a_k,b_k]$,
$\gamma_k(t^{(i)}_k)\to y^{(i)}$, the limits $\lim (a_k-t^{(i)}_k)$,
$\lim (b_k-t^{(i)}_k)$ exist, there are continuous causal curves
$\gamma^{(i)}: [a^{(i)},b^{(i)}] \to M$, $a^{(i)}=\lim
(a_k-t^{(i)}_k)$, $b^{(i)}=\lim (b_k-t^{(i)}_k)$,
$\gamma^{(i)}(0)=y^{(i)}$, such that the sequence $\gamma^{(i)}_k:
[a_k-t^{(i)}_k,b_k-t^{(i)}_k] \to M$, defined by
$\gamma^{(i)}_k(t)=\gamma_k(t+t^{(i)}_k)$, converges $h$-uniformly
on compact subsets to $\gamma^{(i)}$. The curves $\gamma^{(i)}$ are
past inextendible iff $a^{(i)}=-\infty$ and future inextendible iff
$b^{(i)}=+\infty$.

Given $i \ne j$, $i,j \in G$, there are only three possibilities.
Either $t^{(i)}_k-t^{(j)}_k \to c^{(ij)}=-c^{(ji)}=const.$, in which
case we write $(i) \sim (j)$, or $t^{(i)}_k-t^{(j)}_k \to \pm
\infty$. In the first case, $\gamma^{(i)}$ and $\gamma^{(j)}$ are
one the reparameterization of the other,
$\gamma^{(i)}(t)=\gamma^{(j)}(t+c^{(ij)})$. If $t^{(i)}_k-t^{(j)}_k
\to +\infty$ then for every $t^{(i)} \in [a^{(i)},b^{(i)}]$ and
$t^{(j)} \in [a^{(j)},b^{(j)}]$, there is $K(t^{(j)},t^{(i)})$ such
that for $k>K$, $\gamma^{(j)}_k(t^{(j)}) \le
\gamma^{(i)}_k(t^{(i)})$, in particular for every $x^{(j)} \in
\gamma^{(j)}$  and $x^{(i)} \in \gamma^{(i)}$, $(x^{(j)},x^{(i)})
\in \bar{J}^{+}$, and analogously with the roles of $i$ and $j$
inverted if $t^{(i)}_k-t^{(j)}_k \to -\infty$.

Finally, if $\gamma_n$ is limit maximizing then each $\gamma^{(i)}$
is maximizing and \[\sum_{G/\sim} l(\gamma^{(i)}) \le \limsup
l(\gamma_n).\]

\end{itemize}
\end{theorem}

%

\begin{proof}
The starting point is a limit curve lemma \cite{galloway86b}
\cite[Lemma 14.2]{beem96} whose proof is contained in the proof of
\cite[Prop. 3.31]{beem96}.
\begin{quote}
(Limit curve lemma) Let $\gamma_n: (-\infty,+\infty) \to M$,  be a
sequence of inextendible continuous causal curves parametrized with
respect to $h$-length, and suppose that $y\in M$ is an accumulation
point of the sequence $\gamma_n(0)$. There is  a  inextendible
continuous causal curve $\gamma: (-\infty,+\infty)  \to M$, such
that $\gamma(0)=y$ and a subsequence $\gamma_k$ which converges
$h$-uniformly on compact subsets to $\gamma$.
\end{quote}
The concept of $h$-uniform convergence on compact subsets used in
\cite{beem96} is equivalent to that introduced in definition
\ref{vga} if the curves of the sequence and the limit curve are
defined in $(-\infty,+\infty)$. Thus the limit curve lemma holds
also with the notations and definitions of this work. The idea is to
 extend each curve of the sequence into an inextendible
continuous causal curve, apply the limit curve lemma and show that
the limit curve, restricted to a suitable domain, is a limit curve
for the unextended sequence.
\begin{itemize}
\item Proof of (1). Parametrize the sequence with respect to
$h$-length and pass to a subsequence to get $\gamma_n: [a_n,b_n] \to
M$, $0\in [a_n,b_n]$, $\gamma_n(0) \to y$. Pass to a subsequence
$\gamma_i$ so that there are $a\le 0$, and $b\ge 0$, such that $a_i
\to a$, $b_i \to b$. If there is a neighborhood $U$ of $y$ such that
only a finite number of $\gamma_i$ is contained in $U$, then since
$\gamma_i(0)\to y$ their $h$-lengths are bounded from below by a
positive constant $\epsilon>0$, thus $b_i-a_i>\epsilon$ and finally
$b-a>\epsilon$, thus $b > 0$ or $a <0$. Conversely, if $b-a>0$ then
the $h$-lengths of $\gamma_i$ are bounded from below by a positive
constant $\epsilon$. Let $V\ni y$ be the globally hyperbolic
neighborhood of lemma \ref{mas}, if $\gamma_i$ is contained in $V$
then $\epsilon\le l_0(\gamma_i\vert_{[a_i,b_i]}) \le K \vert
x^{0}(\gamma_i(b_i))-x^{0}(\gamma_i(a_i))\vert$. Let $U\subset V$ be
such that the range of function $x^0$ on $\bar{U}$ is less than $
K^{-1}\epsilon/2$, then no $\gamma_i$ can be contained in $U$, and
in particular only a finite number of $\gamma_i$ is contained in
$U$.

Extend arbitrarily the curves (for instance, if $b_i<+\infty$, use
the exponential map at $\gamma_i(b_i)$ to join $\gamma_i$ with a
future inextendible causal geodesic)  so as to obtain a sequence of
inextendible continuous causal curves $\tilde{\gamma}_i:
(-\infty,+\infty) \to M$, parametrized with respect to $h$-length.
 Apply the limit curve
lemma and infer the existence of a continuous causal limit curve
$\tilde{\gamma}: (-\infty,+\infty) \to M$, to which a  subsequence
$\tilde{\gamma}_k$ of $\tilde{\gamma}_i$ converges $h$-uniformly on
compact subsets. Define $\gamma=\tilde{\gamma}\vert_{[a,b]}$, then
it  follows from the definition of $h$-uniform convergence  that the
subsequence $\gamma_k$ of $\gamma_i$ converges $h$-uniformly on
compact subset to $\gamma$.

If $a >-\infty$, since $a_k \to a$, then all but a finite number of
$a_k$ satisfies $a_k>-\infty$, which becomes all passing to a
subsequence if necessary. The sequence $x_k=\gamma_k(a_k)$ converges
to $\gamma(a)$ as a result of theorem \ref{ups} case (2). It is
clear that if all but a finite number of $a_k$ are finite then
$a=-\infty$, however, even if $a_k >-\infty$ it is possible to infer
that necessarily $a=-\infty$. This happens if $x_k=\gamma_k(a_k) \to
+\infty$, because then since $\gamma_k(0) \to y$, $d_0(x_k,y) \to
+\infty$ and hence $-a_k=l_0(\gamma_k\vert_{[a_k,0]}) \to +\infty$.
Also if $l(\gamma_k\vert_{[a_k,0]}) \to +\infty$ then $a=-\infty$,
as it follows from lemma \ref{vja} case (ii). Analogous arguments
hold in the future case.

Next $\tilde{\gamma}$ is future inextendible thus $\gamma$ is future
inextendible if $b=+\infty$ otherwise it has future endpoint. The
fact that $\gamma_k$ are future inextendible iff $b_k=+\infty$
follows from lemma \ref{nsf}. Analogous arguments hold in the past
case. The last statement follows from theorem \ref{maxs}.

\item Proof of (2). Parametrize the sequence with respect to
$h$-length to get $\gamma^x_n: [0,b_n] \to M$, $x_n=\gamma^x_n(0)
\to x$, $z_n=\gamma^x_n(b_n) \to z$. Pass to a subsequence
$\gamma^x_i$ so that there is $b\ge 0$, such that $b_i \to b$. If
there is a neighborhood $U$ of $x$ such that only a finite number of
$\gamma^x_i$ is contained in $U$, then since $\gamma^x_i(0)\to x$
their $h$-lengths are bounded from below by a positive constant
$\epsilon>0$, thus $b_i>\epsilon$ and finally $b>\epsilon>0$.
Conversely, if $b>0$ then the $h$-lengths of $\gamma^x_i$ are
bounded from below by a positive constant $\epsilon$. Let $V\ni x$
be the globally hyperbolic neighborhood of lemma \ref{mas}, if
$\gamma^x_i$ is contained in $V$ then $\epsilon\le
l_0(\gamma^x_i\vert_{[0,b_i]}) \le K \vert
x^{0}(\gamma^x_i(b_i))-x^{0}(\gamma^x_i(0))\vert$. Let $U\subset V$
be such that the range of function $x^0$ on $\bar{U}$ is less than $
K^{-1}\epsilon/2$, then no $\gamma^x_i$ is contained in $U$, and in
particular only a finite number of $\gamma^x_i$ is contained in $U$.
Clearly, if $z \ne x$ then there is $U\ni x$, $z \notin U$, such
that only a finite number of $\gamma^x_i$ is contained in $U$. If
$x_n=z_n$ take $U\ni x$ convex and hence causal, then $\gamma^x_i$
necessarily escape $U$ otherwise they would violate causality.

Extend arbitrarily the curves so as to obtain a sequence of
continuous causal curves $\tilde{\gamma}^x_i: (-\infty,+\infty) \to
M$, parametrized with respect to $h$-length. By lemma \ref{nsf} they
are inextendible. Apply the limit curve lemma and infer the
existence of a continuous causal limit curve $\tilde{\gamma}^x:
(-\infty,+\infty) \to M$, to which a subsequence
$\tilde{\gamma}^x_j$ of $\tilde{\gamma}^x_i$ converges $h$-uniformly
on compact subsets. Define $\gamma^x=\tilde{\gamma}^x\vert_{[0,b]}$,
then it is follows from the definition of $h$-uniform convergence
that the subsequence $\gamma^x_j$ of $\gamma^x_i$ converges
$h$-uniformly on compact subsets to $\gamma^x$.

Repeat the argument for $\gamma^z_j(t)=\gamma^x_j(t+b_j)$, and find
the existence of the continuous causal curve $\gamma^z: [-b,0] \to
M$, to which the subsequence $\gamma^z_k$ of $\gamma^z_j(t)$
converges $h$-uniformly on compact subsets (clearly $\gamma^x_k$
converges $h$-uniformly on compact subset to $\gamma^x$).

If $b<+\infty$, by remark \ref{rem} $\gamma^z_k$ converges
$h$-uniformly to $\gamma^z$. Given $t \in [0,b]$,
$\gamma^x(t)=\lim_{k \to +\infty} \gamma^x_k(t)=\lim_{k \to +\infty}
\gamma_k^z(t-b_k)$, but $t-b_k \to t-b \in [-b,0]$, and using point
(b) of theorem \ref{ups}, $\lim_{k \to +\infty}
\gamma_k^z(t-b_k)=\gamma^z(t-b)$, thus $\gamma^x(t)=\gamma^z(t-b)$,
and $\gamma^x$ and $\gamma^z$ are one the reparameterization of the
other. In particular $\limsup l(\gamma^x_k) \le l(\gamma^x)$ by
theorem \ref{ups}.

If $b=+\infty$ since $\tilde{\gamma}^x$ is future inextendible and
$\gamma^x=\tilde{\gamma}^x\vert_{[0,+\infty)}$ then $\gamma^x$ is
future inextendible. An analogous argument holds for $\gamma^z$,
which can be written $\gamma^z=\tilde{\gamma}^z
\vert_{(-\infty,0]}$, where $\tilde{\gamma}^z$ is past inextendible.
Let $t^x \in [0,+\infty)$, $t^z \in (-\infty,0]$, there is a
constant $K>0$ such that for $k>K$, $b_k>t^x-t^z$, thus
$\gamma^z_k(t^z)=\gamma^x_k(t^z+b_k) \ge \gamma^x_k(t^x)$, in
particular passing to the limit $k \to +\infty$, and using the
pointwise convergence
\begin{equation}
(\gamma^x(t^x),\gamma^z(t^z)) \in \bar{J}^{+}.
\end{equation}
 If $\gamma^x$ is not a lightlike ray
then there is $t \in [0,+\infty)$ such that $\gamma^x(t) \in
I^{+}(x)$, but then for every $t' \in (-\infty,0]$,
$(\gamma^x(t),\gamma^z(t')) \in \bar{J}^{+}$, and since $I^{+}(x)$
is open, $\gamma^z(t') \in \bar{J}^{+}(x)$, and analogously in the
other case. Finally, consider the case in which neither $\gamma^x$
nor $\gamma^z$ are lightlike lines. There are $t^x \in [0,+\infty)$
and $t^z \in (-\infty,0]$ such that $\bar{x}=\gamma^x(t^x) \in
I^{+}(x)$ and $\bar{z}=\gamma^z(t^z) \in I^{-}(z)$, thus since
$(\bar{x},\bar{z}) \in \bar{J}^{+}$ and $I^+$ is open, $z \in
I^{+}(x)$.

It is clear that if $x_n=x$ and $z_n=z$, $d(x,z) \ge \limsup
l(\gamma^x_k)$. The case in which $x_n=x$ and $\gamma^z$ is not a
lightlike ray or the case in which $\gamma^x$ is not a lightlike ray
and $z_n=z$ is simpler than the last case in which both $\gamma^x$
and $\gamma^z$ are not lightlike rays. I am going to give the proof
of $d(x,z) \ge \limsup l(\gamma^x_k)$ in this last case. Define
$\bar{x}_k= \gamma^x_k(t^x)$, $\bar{z}_k=\gamma^z_k(t^z)$. For large
enough $k$, $\bar{x}_k \in I^{+}(x)$, $\bar{z}_k \in I^{-}(z)$ and
$\bar{x}_k\le \bar{z}_k$. The reverse triangle inequality gives
$d(x,z)\ge d(x,\bar{x}_k)+d(\bar{x}_k,\bar{z}_k)+d(\bar{z}_k,z)$. If
$d(x,z)=+\infty$ the inequality $d(x,z) \ge \limsup l(\gamma^x_k)$
is obvious. Assume $d(x,z)<+\infty$, then for sufficiently large
$k$, $d(x,\bar{x}_k)<+\infty$, $d(\bar{x}_k,\bar{z}_k)<+\infty$ and
$d(\bar{z}_k,z)<+\infty$. By the lower semi-continuity of the
distance, $d(x,\bar{x})<+\infty$ otherwise $\liminf
d(x,\bar{x}_k)=+\infty$ and hence $d(x,z)=+\infty$. Analogously, it
is also $d(\bar{z},z)<+\infty$. Note that
\begin{align*}
d(\bar{x}_k,\bar{z}_k) \ge &  \,l(\gamma^x_k\vert_{[t^x,t^z+b_k]})=
l(\gamma^x_k)-l(\gamma^x_k\vert_{[0,t^x]})-l(\gamma^x_k\vert_{[t^z+b_k,b_k]})\\
&\
=l(\gamma^x_k)-l(\gamma^x_k\vert_{[0,t^x]})-l(\gamma^z_k\vert_{[t^z,0]}).
\end{align*}
Given $\epsilon>0$, use the uniform convergence on compact subsets
of $\gamma^x_k$ and $\gamma^z_k$, the upper semi-continuity of the
length functional, and the lower semi-continuity of the distance to
obtain for sufficiently large $k$
\begin{align*}
l(\gamma^x_k\vert_{[0,t^x]})&\le
l(\gamma^x\vert_{[0,t^x]})+\epsilon\le d(x,\bar{x})+\epsilon\le d(x,\bar{x}_k)+2\epsilon,\\
l(\gamma^z_k\vert_{[t^z,0]})&\le
l(\gamma^z\vert_{[t^z,0]})+\epsilon\le d(\bar{z},z)+\epsilon\le
d(\bar{z}_k,z)+2\epsilon.
\end{align*}
Putting everything together gives $d(x,z)\ge
l(\gamma^x_k)-4\epsilon$. Taking the limsup and using the the
arbitrariness of $\epsilon$,  $d(x,z) \ge \limsup l(\gamma^x_k)$. If
the spacetime is non-total imprisoning and $b=+\infty$ then the
future inextendible curve $\gamma^x$ escapes every compact and thus
the curves $\gamma_n$ can't all be contained in a compact.

Assume that $\gamma_k$ are not entirely contained in a compact then
since $\gamma^x_k(0) \to x$, there must be a subsequence whose
$h$-length goes to infinity thus it can't be $b_k \to b<+\infty$,
and hence $b=+\infty$. If $l(\gamma_n) \to +\infty$ then
$l(\gamma^x_k) \to +\infty$ and $b=+\infty$ follows from point (ii)
of lemma \ref{vja}.

If $x=z$ and $b=+\infty$ then $\gamma=\gamma^x \circ
\gamma^z:(-\infty,+\infty) \to M$ is clearly inextendible and it is
a limit curve of $\gamma_n$ because both $\gamma^x$ and $\gamma^z$
are limit curves. Given $t_1,t_2\in\mathbb{R}$, let
$\tilde{x}=\gamma(t_1)$ and $\tilde{z}=\gamma(t_2)$. If $t_1 \le
t_2$, clearly $(\tilde{x},\tilde{z})\in J^{+}\subset \bar{J}^{+}$.
If $t_1>t_2$, $t_1\ge 0$ and $t_2 \le 0$ then $\tilde{x} \in
\gamma^x$, $\tilde{z}\in \gamma^z$, and it has been already
established that $(\tilde{x},\tilde{z})\in \bar{J}^{+}$. There
remain the cases $t_1>t_2$, $t_1,t_2\le 0$ and $t_1>t_2$,
$t_1,t_2\ge 0$. I consider the former case the other being similar.
Take $t'_1 >0$, $\tilde{x}'=\gamma(t_1')$, so that
$(\tilde{x}',\tilde{z}) \in \bar{J}^{+}$. For every $w \in
I^{-}(\tilde{x})$, since $(\tilde{x},\tilde{x}') \in J^{+}$, it is
$(w,\tilde{x}') \in I^{+}$, but $I^{+}$ is open, thus $\tilde{z} \in
\bar{J}^{+}(w)$ and since $w$  can be taken arbitrarily close to
$\tilde{x}$, $(\tilde{x},\tilde{z}) \in \bar{J}^{+}$.

In particular, given $\tilde{x} \in \gamma$, take $\tilde{z} \in
\gamma$, $\tilde{z}\ne \tilde{x}$, such that $(\tilde{x},\tilde{z})
\in J^{+}$, but then it is also $(\tilde{z},\tilde{x}) \in
\bar{J}^{+}$ which implies that strong causality does not hold at
$\tilde{x}$ (see, for instance, theorem 3.4 of \cite{minguzzi07b}).

If $\gamma$ is not a lightlike line then there are $t^z \in
(-\infty,0]$ and $t^x \in [0,+\infty)$ such that $\gamma(t^z) \ll
\gamma(t^x)$, which reads $\gamma^z(t^z) \ll \gamma^x(t^x)$ and for
sufficiently large $k$, $\gamma^z_k(t^z) \ll \gamma^x_k(t^x)$. It
has been proved that for $k
>K(t^x,t^z)$, it is $b_k>t^x-t^z$. Consider the continuous causal curve
$\gamma^x_k\vert_{[t^x,t^z+b_k]}$ of endpoints $\gamma^x_k(t^x)$ and
$\gamma^x_k(t^z+b_k)=\gamma^z_k(t^z)$. Clearly,
$y=\gamma^x_k(t^x+b_k/2)$ is such that $y\ll y$ which is impossible
if the curves $\gamma_k^x$ do not intersect the chronology violating
set of $(M,g)$.


If $b<+\infty$ then $\gamma^x:[0,b] \to M$, is such that
$\gamma^x(0)=\gamma^x(b)$ thus making infinite rounds over
$\gamma^x$ it is possible to obtain an inextendible continuous
causal curve $\gamma$ passing infinitely often through $x$. If
$\gamma$ is not a lightlike line then it is contained in the
chronology violating se by theorem \ref{xxa}.

If $\gamma_n$ is limit maximizing then $\gamma^x_k$ and $\gamma^z_k$
are limit maximizing and thus $\gamma^x$ and $\gamma^z$ are
maximizing (theorem \ref{maxs}). If $b=+\infty$, by the same
theorem, given $t^x \in [0,+\infty)$ and $t^z \in (-\infty,0]$
\begin{align*}
\lim_{k \to +\infty} &l(\gamma^x_k\vert_{[0,t^x]})= l(\gamma^x
\vert_{[0,t^x]}),\\
 \lim_{k \to +\infty} &
l(\gamma^z_k\vert_{[t^z,0]})= l(\gamma^z \vert_{[t^z,0]}).
\end{align*}
Thus given $\epsilon>0$, it is for sufficiently large $k$,
\[l(\gamma^x \vert_{[0,t^x]})+l(\gamma^z \vert_{[t^z,0]})\le
l(\gamma^x_k\vert_{[0,t^x]})+l(\gamma^z_k\vert_{[t^z,0]})+\epsilon.\]
Note that $\gamma^z_k\vert_{[t^z,0]}$ is a reparameterization of
$\gamma^x_k\vert_{[t^z+b_k,b_k]}$. But for sufficiently large $k$,
$b_k>t^x-t^z$ and thus $[0,t^x]\cap[t^z+b_k,b_k]=\emptyset$. As a
consequence,
$l(\gamma^x_k\vert_{[0,t^x]})+l(\gamma^z_k\vert_{[t^z,0]})\le
l(\gamma^x_k)$, and taking the limit
\[
l(\gamma^x \vert_{[0,t^x]})+l(\gamma^z \vert_{[t^z,0]})\le \limsup
l(\gamma_n)+\epsilon
\]
From the arbitrariness of $t^x$, $t^y$ and $\epsilon$ the thesis
follows. The last statement is obvious.

\item Proof of (3). If $\liminf(b_n-a_n)>0$ it is possible to find a subsequence $\gamma_i$ such that
$\lim(b_i-a_i)$ exists and is greater than zero. The subsequence can
be chosen so that $\gamma_i(t^{(1)}_i) \to y^{(1)}$ for suitable
$t^{(1)}_i \in [a_i,b_i]$. It is also possible to assume that $\lim
(a_i-t_i^{(1)})$ and $\lim (b_i-t^{(1)}_i)$ exist (otherwise pass to
another subsequence denoted in the same way).  Extend arbitrarily
the curves to get inextendible continuous causal curves
$\tilde{\gamma}_i$. Translate their domain to get a sequence
$\tilde{\gamma}^{(1)}_i(t)=\tilde{\gamma}_i(t+t^{(1)}_i)$ and apply
the limit curve lemma to infer the existence of
$\tilde{\gamma}^{(1)}$ to which a subsequence
$\tilde{\gamma}^{(1)}_{k^{(1)}}$ of $\tilde{\gamma}^{(1)}_i$
converges uniformly on compact subsets. Define
\[\gamma^{(1)}=\tilde{\gamma}^{(1)}\vert_{[\lim(a_{k^{(1)}}-t^{(1)}_{k^{(1)}}),
\lim(b_{k^{(1)}}-t^{(1)}_{k^{(1)}} )]},\] then the sequence
$\gamma^{(1)}_{k^{(1)}}(t)=\gamma_{k^{(1)}}(t+t^{(1)}_{k^{(1)}})$
obtained translating the domains of the subsequence
$\gamma_{k^{(1)}}$ of $\gamma_n$, converges $h$-uniformly on compact
subsets to $\gamma^{(1)}$. Note that since $\gamma_i(t^{(1)}_i) \to
y^{(1)}$ it is $\gamma_{k^{(1)}}(t^{(1)}_{k^{(1)}}) \to y^{(1)}$.

Repeat the argument for $y^{(2)}$, but this time starting from
$\gamma_{k^{(1)}}$ instead from $\gamma_n$. The result is the
existence of a subsequence $\gamma_{k^{(2)}}$ of $\gamma_{k^{(1)}}$,
and of a sequence $t^{(2)}_{k^{(2)}}\in [a_{k^{(2)}},b_{k^{(2)}}]$
such that the limits $\lim(a_{k^{(2)}}-t^{(2)}_{k^{(2)}})$,
$\lim(b_{k^{(2)}}-t^{(2)}_{k^{(2)}} )$ exist,
$\gamma_{k^{(2)}}(t^{(2)}_{k^{(2)}}) \to y^{(2)}$, and the
translated subsequence
$\gamma_{k^{(2)}}^{(2)}(t)=\gamma_{k^{(2)}}(t+t^{(2)}_{k^{(2)}})$
converges $h$-uniformly on compact subsets to a continuous causal
curve $\gamma^{(2)}$ of domain
$[\lim(a_{k^{(2)}}-t^{(2)}_{k^{(2)}}),
\lim(b_{k^{(2)}}-t^{(2)}_{k^{(2)}} )]$.

Continue in this way for every $(i)$ so as to obtain a sequence of
subsequences of $\gamma_n$: $\gamma_{k^{(1)}}$, $\gamma_{k^{(2)}}$,
$\ldots$ with analogous properties. Apply a Cantor diagonal process,
namely, construct the new sequence $\gamma_k$ as follows. Define
$\gamma_1$ to be the first  curve of $\gamma_{k^{(1)}}$, define
$\gamma_2$ to be the second curve of $\gamma_{k^{(2)}}$, and so on.
In this way $\gamma_k$ is a subsequence of  $\gamma_{k^{(i)}}$ for
every $(i)$.

All the other statements of case (3) have analogous proofs in (1) or
(2). For instance,  the inequality $\sum^{\infty}_i l(\gamma^{(i)})
\le \limsup l(\gamma_n)$, follows from  $\sum^{N}_i l(\gamma^{(i)})
\le \limsup l(\gamma_n)$ which is proved in a way completely
analogous to case (2). In the last step  the arbitrariness of $N$ is
used.

\end{itemize}
\end{proof}
%
%
%
%


\section{Some consequences}

In this section some unpublished consequences of the limit curve
theorem are explored.

\begin{remark}
Observe that given a causal curve $\gamma:I \to M$  the function
$s(t_1,t_2)=d(\gamma(t_1),\gamma(t_2))-l(\gamma \vert_{[t_1,t_2]})$,
$s: I\times I \to \mathbb{R}$, is, by the reverse triangle
inequality, non-decreasing in $t_2$ and non-increasing in $t_1$ and
it is lower semi-continuous in $t_1$ and $t_2$. As a consequence if
$s(t_1,t_2)>0$ then the same is true in a neighborhood of $(t_1,t_2)
\in I\times I$. In particular if $\gamma$ is maximizing in $(a',b')
\subset I$, then there is a (non unique) closed maximal interval
$[a,b]\supset (a',b')$ in which $\gamma$ is maximizing ($a$ and $b$
can be infinite). Therefore it is natural to assume that the domain
of definition of a maximizing causal curve is a closed set. Indeed,
if not it can be prolonged as a geodesic to get a maximizing causal
curve defined on a closed set.
\end{remark}

Newman \cite{newman90} \cite[Prop. 4.40]{beem96} proved that given a
maximizing timelike segment $c: [a,b] \to M$, the spacetime $(M,g)$
is strongly causal at $(a,b)$.
Here a  short proof is given. Moreover, the result is extended to
maximizing lightlike segments.

A lemma is needed (it generalizes \cite[Lemma 4.38]{beem96})

\begin{lemma} \label{mod}
Let $c: [a,b] \to M$ be a maximizing causal curve   on the spacetime
$(M,g)$, then either causality holds at $c(t)$ for every $t \in
[a,b]$ or $c$ can be obtained from the domain restriction of a
closed lightlike line $\gamma$ (closed here means that there are
$t_1$ and $t_2$ in the domain of $\gamma$ such that
$\gamma(t_1)=\gamma(t_2)$ and $\gamma'(t_1) \propto \gamma'(t_2)$).
\end{lemma}

\begin{proof}

Assume causality does not hold at $x=c(t)$, $t \in [a,b]$, then
there is $z$ such that $x<z<x$. Let $\eta_1:[d,e] \to M$ be the
causal curve connecting $x$ to $z$ and let $\eta_2:[e,f] \to M$ be
the causal curve connecting $z$ to $x$. Define $\eta=\eta_2 \circ
\eta_1$. Since $c$ is maximizing, $x \notin I^{+}(x)$, and thus the
causal curve $\eta: [d,f] \to M$  which connects $x$ to $z$ and then
$z$ to $x$ is a maximizing lightlike geodesic (see theorem
\ref{xxa}) with $\eta'(d) \propto \eta'(f)$. Assume that $\eta'(d)$
is not proportional to ${c}'(t)$. If $t\ne a$ it is $d(c(a),z)
> l(c\vert_{[a,t]})$ because $\eta_1\circ c\vert_{[a,t]}$ connects
$c(a)$ to $z$  and has length $l(c\vert_{[a,t]})$ but it is not a
geodesic and hence it is not maximizing. If $t\ne b$ it is
$d(z,c(b))
> l(c\vert_{[t,b]})$ because $ c\vert_{[t,b]}\circ \eta_2$ connects
 $z$ to $c(b)$  and has length $l(c\vert_{[t,b]})$ but it is not a geodesic
and hence it is not maximizing.
For every $t \in[a,b]$
\[
d(c(a),c(b))\ge d(c(a),z)+ d(z,c(b))>
l(c\vert_{[a,t]})+l(c\vert_{[t,b]})=l(c) \] and $c$ would not be
maximizing. Thus $\eta'(d)$ ( $=\eta'(f)$) is  proportional to
${c}'(t)$ which implies that making many rounds over $\eta$ a
inextendible maximizing curve $\gamma$ can be obtained. It also
implies, since the solution to the geodesic equation is unique, that
$c$ is obtained from the suitably parametrized curve $\gamma$
through a restriction of its domain.
\end{proof}

\begin{theorem} \label{str}
Let $c: [a,b] \to M$ be a maximizing causal curve on the spacetime
$(M,g)$, then there are the following possibilities
\begin{itemize}
\item[(i)]  $c$ is timelike and $(M,g)$ is strongly causal at $c(t)$ for every
$t \in (a,b)$.
\item[(ii)]  $c$ is lightlike and one of the following
possibilities holds
\begin{enumerate}
\item $(M,g)$ is strongly causal at   $c(t)$ for every
$t \in (a,b)$.
\item
Strong causality is violated at every point of $c$, and given
$t_1,t_2 \in [a,b]$, $t_1<t_2$, it is $(c(t_2),c(t_1)) \in
\bar{J}^{+}$.
\item $c$ intersects the closure of the chronology violating set at some point $x=c(t)$,
$t \in (a,b)$. Moreover, all the points in $c((a,b))$ at which
strong causality is violated belong to the closure of the chronology
violating set. In particular, $(M,g)$ is not chronological.
\end{enumerate}
\end{itemize}

\end{theorem}

\begin{proof}
Assume that strong causality fails at some point  $x=c(t)$, $t \in
(a,b)$, otherwise (i) or (ii1) hold. I am going to show that case
(ii2) applies or $x$ belongs to the closure of the chronology
violating set.

Since strong causality fails at $x$ there is a neighborhood $U\ni
x$, and a sequence of causal curves $\gamma_n$ not entirely
contained in $U$ and respectively of endpoints $x_n$, $z_n$, $x_n
<z_n$, $x_n,z_n \to x$. Consider the limit  curve theorem, case (2).
If the two limit curves $\gamma^x$ and $\gamma^z$ are one the
reparameterization of the other, then there is a closed causal curve
starting and ending at $x$, which implies by lemma \ref{mod} that
$c$ is lightlike and can be prolonged  to a closed lightlike line,
in particular (ii2) is verified. If the two limit curves $\gamma^x$
and $\gamma^z$ are not one the reparameterization of each other
then $\gamma=\gamma^x\circ \gamma^z$ is a inextendible continuous
causal curve.

The curve $c$ is either lightlike or timelike. Let me consider for a
moment the former case. If there is some point $y \in \gamma^x  \cap
I^{+}(x)$, then since $(y,x) \in \bar{J}^{+}$ by the limit  curve
theorem, case (2), it is possible to construct a closed timelike
curve passing arbitrarily close to $x$. Thus if $\gamma^x$ is not a
lightlike ray then $x$ belongs to the closure of the chronology
violating set (that is point 3 of the theorem). The same conclusion
holds if $\gamma^z$ is not a lightlike ray. If both $\gamma^x$ and
$\gamma^z$ are lightlike rays but they are not tangent at $x$ then
rounding the corner at $x$ and using again the fact that for every
$q \in \gamma^x$ and $p\in \gamma^z$, $(q,p) \in \bar{J}^+$ it is
possible to construct a closed timelike curve passing arbitrary
close to $x$. Thus again point 3 of the theorem holds. There remains
the possibility in which $\gamma^x$ and $\gamma^z$ are both
lightlike rays tangent at $x$ so that $\gamma$ is a inextendible
lightlike geodesic (not necessarily a line). If $\gamma$ and $c$ are
tangent at $x$ then $c$ is a segment of $\gamma$ and recalling the
properties of $\gamma$ given by the limit curve theorem, case (2),
it follows that point (ii2) of this theorem holds.

There remain two possibilities: (a) $c$ lightlike and $\gamma$ is a
 inextendible lightlike geodesic not tangent to $c$ at $x$ and (b)
$c$ timelike and $\gamma$ is a  inextendible continuous causal curve
passing through $x$. Now the argument is the same for both cases.


Take $\bar{x}\in \gamma^x\backslash\{x\}$ and $\bar{z} \in
\gamma^z\backslash\{x\}$. By the limit curve theorem, case (2),
$(\bar{x},\bar{z})\in \bar{J}^{+}$. We are going to prove that
$d(c(a),\bar{x})>l(c\vert_{[a,t]})$ and in particular
$(c(a),\bar{x}) \in I^{+}$. The inequality has a different proof if
$c$ is lightlike or timelike. If $c$ is lightlike the causal curve
which connects $c(a)$ to $x$ along $c$ and then $x$ to $\bar{x}$
along $\gamma$ is not a geodesic (because they are not tangent at
$x$ and hence there is a corner between the two segments), thus it
is not maximizing, but it has length $l(c\vert_{[a,t]})$ and hence
$d(c(a),\bar{x})\ge l(c\vert_{[a,t]})+\delta_x$ for a certain
$\delta_x>0$. If $c$ is timelike the argument goes as with $c$
lightlike but one has to consider also the possibility that the
segment of $\gamma$ from $x$ to $\bar{x}$ could be a timelike
geodesic prolongation of $c$. However, in this case the same segment
would have length greater than zero and hence again
$d(c(a),\bar{x})\ge l(c\vert_{[a,t]})+\delta_x$. An analogous
argument gives $d(\bar{z},c(b))\ge l(c\vert_{[t,b]})+\delta_z$ and
$(\bar{z},c(b)) \in I^{+}$. The Lorentzian distance is lower
semi-continuous thus given $\epsilon
>0$, such that $d(c(a),\bar{x})-\epsilon>0$ and $d(\bar{z},c(b))-\epsilon>0$, there is a neighborhood
$U\ni \bar{x}$ such that $d(c(a),x')\ge d(c(a),\bar{x})-\epsilon>0$
for $x' \in U$. Analogously, there is a neighborhood $V \ni \bar{z}$
such that $d(z',c(b))\ge d(\bar{z},c(b))-\epsilon>0$ for $z' \in V$.
But $(\bar{x},\bar{z}) \in \bar{J}^{+}$ and thus there is a choice
of $x' \in U$, $z' \in V$ such that $(x',z') \in J^{+}$. The reverse
triangle inequality for $c(a)<x'<z'<c(b)$ gives
\[
d(c(a),c(b)) \ge d(c(a), x')+d(z',c(b))\ge
d(c(a),\bar{x})+d(\bar{z},c(b))-2\epsilon \ge
l(c)+\delta_x+\delta_z-2\epsilon.
\]
Since $\epsilon$ is arbitrary $d(c(a),c(b))\ge
l(c)+\delta_x+\delta_z$, and finally $d(c(a),c(b))> l(c)$.

The contradiction proves that none of the two cases (a) and (b)
applies.

\end{proof}

The previous theorem case (ii) in short states that outside the
closure of the chronology violating set,  the maximizing lightlike
segments propagate the property of strong causality. The next result
will plays role in the study of totally imprisoned curves
\cite{minguzzi07f}.

\begin{figure}[ht]
\begin{center}
 \includegraphics[width=8cm]{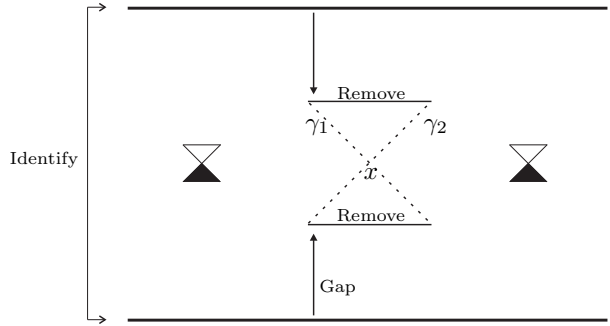}
\end{center}
\caption{A non-totally vicious non-chronological spacetime obtained
from Minkowski spacetime by removing two spacelike segments and by
identifying two lines as displayed. If the gap is large enough there
is only one component of the chronology violating set and no
lightlike line though $\gamma_1$ and $\gamma_2$ are inextendible
lightlike geodesics. If the gap is small enough there are two
components, the interxection of their boundaries being the event
$x$. The curves $\gamma_1$ and $\gamma_2$ are in this case lightlike
lines as expected from theorem \ref{miv}} \label{clnonvicious}
\end{figure}

%
%

\begin{theorem} \label{pku}
Let $\gamma_1: \mathbb{R} \to M$ and $\gamma_2: \mathbb{R} \to M$ be
two distinct lightlike lines which intersect at a event $x$, and
assume that neither $\gamma_1$ nor $\gamma_2$ intersect the closure
of the chronology violating set, then strong causality holds at
every point of $\gamma_1$ and $\gamma_2$.
\end{theorem}

\begin{proof}
Without loss of generality assume $\gamma_1(0)=\gamma_2(0)=x$. By
theorem \ref{str} either strong causality is satisfied at $x$ and
hence it is satisfied at every point of $\gamma_1$ and $\gamma_2$ or
it is not satisfied at $x$ and hence it fails at every point of
$\gamma_1$ and $\gamma_2$ (case (ii3) of theorem \ref{str} can not
hold in this case otherwise $x$ would belong to the closure of the
chronology violating set). In the latter case, split $\gamma_1$ in
two parts $\gamma_1^+: [0,+\infty) \to M$, and $\gamma_1^-:
(-\infty,0] \to M$ and do the same with $\gamma_2$. By the same
theorem for every $\epsilon>0$, defined $x_1^+=\gamma_1(\epsilon)$,
$x_1^-=\gamma_1(-\epsilon)$, $x_2^+=\gamma_2(\epsilon)$,
$x_2^-=\gamma_2(-\epsilon)$, it is $(x_2^+, x_2^-)\in \bar{J}^+$.

Since the lines are distinct it is possible to round  the corner at
$x$ of $\gamma^+_2\circ \gamma_1^-$, and it is possible to do the
same for $\gamma^+_1\circ \gamma_2^-$. As a result, $(x_1^-,
x_2^+)\in I^+$ and $(x_2^-, x_1^+)\in I^+$, and using the fact that
$I^+$ is open it follows that $(x_1^-,x_1^+) \in I^{+}$, a
contradiction since $\gamma_1$ is a line.
\end{proof}

Clearly, the previous theorem holds even if the the lightlike lines
are replaced with lightlike maximizing segments which intersect at a
event which corresponds to parameter values which stay in the
interior of their respective domains.

It is well known since the work of Brandon Carter that the
chronology violating set is the union of open components, where two
points $p$ and $q$ belong to the same component whenever $p\ll q\ll
p$. In some cases the boundaries of the components may intersect
(see figure \ref{clnonvicious}).

\begin{theorem} \label{miv}
Let $B_1=\dot{C}_1$ and $B_2=\dot{C}_2$ where $C_1$ and $C_2$ are
distinct components of the chronology violating set of $(M,g)$.
Through every point  of $B_1\cap B_2$ (which may be empty) there
passes a lightlike line entirely contained in $B_1\cup B_2$. In
particular a spacetime without lightlike lines has a chronology
violating set with components having disjoint closures.
\end{theorem}

\begin{proof}
Assume that $x \in B_1\cap B_2$, and let $\gamma^{(1)}_n$ be a
sequence of closed timelike curves of starting point and ending
point $x^{(1)}_n$, $x^{(1)}_n \in C_1$ and $x^{(1)}_n \to x$ (thus
$\gamma^{(1)}_n$ are contained in $C_1$). Analogously, let
$\gamma^{(2)}_n$ be a sequence of closed timelike curves of starting
point and ending point $x^{(2)}_n$, $x^{(2)}_n \in C_2$ and
$x^{(2)}_n \to x$ (thus $\gamma^{(2)}_n$ are contained in $C_2$).
Apply the limit curve theorem case (2) to $\gamma^{(1)}_n$ with
$z=x$. If $b<+\infty$ then $\gamma^{(1)x}$ is a closed causal curve,
it must be achronal since if $p,q \in \gamma^{(1)x}$, $p \ll q$,
then $x\le p\ll q \le x$ and hence $x\ll x$ which implies $x\in C_1$
a contradiction. Thus $\gamma^{(1)x}$ is a geodesic with no
discontinuity in the tangent vectors at $x$. It can be extended to a
lightlike line $\gamma$ by making infinite  rounds over
$\gamma^{(1)x}$. Moreover, note that $\gamma^{(1)x}$ can't have any
point $p$ in $C_1$ otherwise $x\le p \ll p \le x$ and again $x \ll
x$, $x \in C_1$, a contradiction, thus $\gamma \subset \dot{C}_1$.

It remains to consider the case in which the limit curve theorem
case (2) applies to $\gamma^{(1)}_n$ and $\gamma^{(2)}_n$ with
$b=+\infty$. With the notations of the limit curve theorem, there
are  future inextendible continuous causal curves
$\gamma^{(1)x}\subset \bar{C}_1$, $\gamma^{(2)x}\subset \bar{C}_2$
and past inextendible continuous causal curves $\gamma^{(1)z}\subset
\bar{C}_1$, $\gamma^{(2)z}\subset \bar{C}_2$. Assume that neither
$\gamma^{(21)}=\gamma^{(1)x}\circ\gamma^{(2)z}$ nor
$\gamma^{(12)}=\gamma^{(2)x}\circ\gamma^{(1)z}$ are lightlike lines.
There are points $\bar{z}^{(2)} \in \gamma^{(2)z}$, $\bar{x}^{(1)}
\in \gamma^{(1)x}$, $ \bar{z}^{(2)} \ll \bar{x}^{(1)}$, and
$\bar{z}^{(1)} \in \gamma^{(1)z}$, $\bar{x}^{(2)} \in
\gamma^{(2)x}$, $ \bar{z}^{(1)} \ll \bar{x}^{(2)}$. Since $I^{+}$ is
open there are open sets $V^{(2)} \ni \bar{z}^{(2)}$, $U^{(1)} \ni
\bar{x}^{(1)}$, $V^{(1)} \ni \bar{z}^{(1)}$ and $U^{(2)} \ni
\bar{x}^{(2)}$ such that $V^{(2)}\times U^{(1)} \subset I^{+}$,
$V^{(1)}\times U^{(2)} \subset I^{+}$. But $(U^{(1)}\cap C_1) \times
(V^{(1)}\cap C_1) \subset I^{+}$ because any two points of $C_1$ are
chronologically related. Analogously, $(U^{(2)}\cap C_2) \times
(V^{(2)}\cap C_2) \subset I^{+}$. These relations prove that it is
possible to find two points $p\in C_1$ and $q \in C_2$ such that
$p\ll q\ll p$ and thus the two components would coincide. The
contradiction proves that $\gamma^{(21)}$ or $\gamma^{(12)}$ is a
lightlike line. Assume it is the former the latter case being
analogous. The curve $\gamma^{(1)x}$ can't intersect $C_1$ otherwise
taken any two points of $\gamma^{(1)x}$ in $C_1$ they would be
chronologically related in contradiction with the achronality of
$\gamma^{(21)}$. Thus $\gamma^{(1)x} \subset \dot{C}_1$ and
analogously, $\gamma^{(2)z} \subset \dot{C}_2$, in conclusion
$\gamma^{(21)} \subset B_1 \cup B_2$.
\end{proof}

\section{Conclusions}

In the first sections of this work the topic of limit curve theorems
in Lorentzian geometry has been reviewed. Since the aim was the
formulation of a limit curve theorem which holds even in the case of
curves with endpoints, the definition of uniform convergence on
compact subsets has been generalized to the case in which the
converging curves do not have the same domain of definition.

The upper semi-continuity of the length functional has been given a
new proof which is suitable to this generalized circumstance
(theorem \ref{ups}). It avoids any mention to the property of strong
causality while it replaces $C^0$ convergence with uniform converge.

The notion of limit maximizing sequence has been generalized to the
case of curves without endpoints, as well as the theorem that the
uniform limit of a limit maximizing sequence is a maximizing curve
(theorem \ref{maxs}).

The central result of the work, theorem \ref{main}, is separated
into three parts.

Point (1) gives the generalization of Beem et al. limit curve lemma
to the case of curves with endpoints,  with a few more observations
which are helpful in order to establish whether the limit curve is
inextendible or not. Thanks to the fact that it holds for  curves
with endpoints, it can be used to construct limit maximizing
sequences where one or both endpoints go to infinity. Thus it is
useful in order to establish the existence of lines or rays passing
through a point.

Point (2) focuses on the case in which there have been given, since
the beginning, two limit events $x$ and $z$ and not only one as in
(1). This case is very useful in applications, especially when it
comes to prove the connectedness of spacetime through maximizing
geodesics or similar results. It also becomes specially interesting
when $x$ and $z$  coincide. In this case it provides information on
the existence of lightlike lines given suitable causality violations
on the spacetime. I will consider these aspects in a related work
. If the sequence of curves is limit maximizing, point (2) gives a
bound to the sum of the lengths of the two limit curves,
generalizing a key observation by Newman \cite{newman90}.

Point (3) focuses on a very general case in which the limit points
given in the beginning can be infinite but numerable. In this case
it is proved that through each one of them there passes a uniform
limit curve and in case of a limit maximizing sequence, an upper
bound to the sum of their lengths has been given.

Some consequences of the limit curve theorem have been considered in
the last sections. It has been proved for instance that in
chronological spacetimes the maximizing lightlike segments defined
over open intervals are such that strong causality either is
everywhere violated or everywhere verified over the curve (theorem
\ref{str}). A consequence is that if in a chronological spacetime
two distinct lightlike lines intersect each other then strong
causality holds at the points of their union (theorem \ref{pku}).

Finally, the last result has been the proof that any spacetime
without lightlike lines has a chronology violating set such that the
closures of its components are disjoint.


\section*{Acknowledgments}

Useful conversations with M. S\'anchez on the upper semi-continuity
of the length functional are acknowledged. This work has been
partially supported by GNFM of INDAM and by MIUR under project PRIN
2005 from Universit\`a di Camerino.


\end{document}